\documentclass{article}

\usepackage{arxiv}

\usepackage[utf8]{inputenc} % allow utf-8 input
\usepackage[T1]{fontenc}    % use 8-bit T1 fonts
\usepackage{hyperref}       % hyperlinks
\usepackage{url}            % simple URL typesetting
\usepackage{booktabs}       % professional-quality tables
\usepackage{amsfonts}       % blackboard math symbols
\usepackage{nicefrac}       % compact symbols for 1/2, etc.
\usepackage{microtype}      % microtypography
\usepackage{lipsum}
\usepackage{amsmath}
\usepackage{amssymb}
\usepackage{epsfig,cite}
\usepackage{yfonts}
\usepackage{bbold}
\usepackage{amsthm}
\newtheorem{theorem}{Theorem}

\newtheorem{lemma}{Lemma}
\title{Absence of eigenvalues in the continuous spectrum for Klein-Gordon operators}

\author{
  R. Ferreira \\
  Universidade de S\~ao Paulo -- USP \\
  S\~ao Paulo, SP, 05508-090, Brazil\\
  \texttt{rafaelufpi@gmail.com} \\
  \And 
   F. N. Lima \thanks{I am corresponding author.} \\
  GTMCOQ, Instituto Federal do Piau\'i -- IFPI \\
  São Raimundo Nonato, Piau\'i, 64770-000, Brazil\\
   \texttt{nogueira@ifpi.edu.br}\\
   \And 
  A. S. Ribeiro \\
  Instituto Federal do Piau\'i -- IFPI \\
  São Raimundo Nonato, Piau\'i, 64770-000, Brazil\\
  \texttt{antonius.ribeirus@ifpi.edu.br} \\
}

\begin{document}
\maketitle

\begin{abstract}
We construct the one-dimensional analogous of \textit{von}-Neumann Wigner potential to the relativistic Klein-Gordon operator, in which is defined taking asymptotic mathematical rules in order to obtain existence conditions of eigenvalues embedded in the continuous spectrum.  Using our constructed potential, we provide an explicit and analytical example of the Klein-Gordon operator with positive eigenvalues embedded in the so called relativistic "continuum region". Even so in this not standard example, we present the region of the "continuum" where those eigenvalues cannot occur. Besides, the absence of eigenvalues in the continuous spectrum for Klein-Gordon operators is proven to a broad general potential classes, including the minimally coupled electric Coulomb potential. Considering known techniques available in literature for Schrodinger operators, we demonstrate an expression for Klein-Gordon operator written in Schrodinger's form, whereby is determined the mathematical spectrum region of absence of eigenvalues.
\end{abstract}


\section{Introduction\label{sec1}}

Theories involving general spectral problems are well-developed of a mathematical point of view for operators from non-relativistic quantum mechanics. However, the known relativistic operators that arising in the relativistic quantum mechanics still require certain mathematical rigours in terms of indispensable properties that need demonstrated in spectral theory. As example, since of \textit{von} Newmman-Wigner work \cite{von1993merkwurdige}, we have know that there is a potential whereby eigenvalues equations for Schrodinger operator provide an unitary positive energy \cite{ReedSimon4}. Great efforts of scientific community are aimed to investigate potential classes in which eigenvalues equations provide only negative eigenvalues. In particular, in the non-relativistic case of Schrodinger there are general theorems that proven the mathematical and physical question of non-existence of positive eigenvalues in the continuum or essential spectrum, see Refs. \cite{Kato1959, Odeh1965, Weidmann2010, Agmon1969, Agmon1970, Simon1969}. For Schrodinger's operator the positive spectrum is known as continuum spectrum, and a relevant problem of a mathematical point of view it is localize the essential spectrum that for Schrodinger's operator is exactly the same continuum spectrum.  

A non-trivial question in mathematical physics is localize the essential spectrum of the relativistic Klein-Gordon (K-G) operator coupled to some interaction potentials. There is a consistent demonstration for some potential classes that proven that positive relativistic region of the continuum (also namely in literature of positive continuum region \cite{VorGitTyu2007}), it is the essential spectrum, that is, the region $[m,\infty)$ \cite{Weder1977,Weder1978}. The negative continuum region occurs in $(-\infty,-m]$, being known as Dirac's sea. Conditions of absence of eigenvalues for Schrodinger's operators are well addressed problems in literature, as discussed above (more details can be seen in Refs.  \cite{Kato1959, Odeh1965, Weidmann2010, Agmon1969, Agmon1970, Simon1969}). However, such conditions for relativistic operators still need of general theorems, i. e., theorems capable to define with mathematical rigours the absence of eigenvalues in the located essential spectrum. In particular,  in physical problems the continuum spectrum is traditionally taken (without lost of generality), and locating the essential spectrum is commonly treated as an essentially mathematical problem. In fact, by use of the relativistic Virial theorem, using the version demonstrated by Kalf \cite{Kalf1976}, it is possible proving the theorem of absence of eigenvalues in systems with Coulomb interaction potential to Dirac's operator.  An alternative explanation can be found in Ref. \cite{Thaller1992}. 

For Klein-Gordon (K-G) operators, general theorems of absence of eigenvalues keep non-demonstrated up to now. A standard reduction procedure can be applied to Schrödinger and Dirac operators, where there is an inner product that allows reductions to partial wave subspace problems, see for example Refs. \cite{Thaller1992, Gitman2012, Oliveira2009}. This does not occurs to K-G operators that have a non-trivial structure of "energy norm", and probably being one of the reasons for the lack of such standard results. Weyl's criterion or the asymmetry form method (see ref. \cite{Gitman2012}) for self-adjointness in the half-line cannot be used for K-G operators due to the reduction procedures at each quantum number channel represent a non-trivial problem, as already mentioned. As discussed methodologies (Weyl's criterion or the asymmetry form method) are not applied in this case, the self-adjointness property cannot be extend to whole operator in $\mathbb{R}^{3}$. In particular, even in the level of self-adjointness theorems, only a few number of works (especially taking Coulomb singularities) are available in literature to K-G operators, as discussed by Gitman et al \cite{Gitman2012}. Although, Coulomb's potentials have been suitable to apply Weyl's criterion. We have know that for Schrödinger's operators there are physically reasonable examples of potentials providing bound states with positive energy. The first potential presented in literature is proposed by Wigner and \textit{von} Newmman \cite{von1993merkwurdige}, with its experimental evidence demonstrated in Ref.  \cite{Capasso1992}. In the ordinary framework of the relativistic quantum mechanics there are no potentials providing eigenvalues embedded in the continuum spectral region of relativistic operators. 

We provide an interaction potential in which the K-G operator has eigenvalues embedded in the continuum region, i. e., in $(- \infty,-m] \bigcup [m, \infty )$. General results on the absence of eigenvalues in the essential spectrum to some self-adjoint K-G operators obtained by Weder (see Refs. \cite{Weder1977, Weder1978} are presented taking electric Coulomb interactions, the core of this work. We have used the same notation introduced by Schechter \cite{Schechter1971}, and the same formalism extensively taken in Weder's works, Refs. \cite{Weder1977, Weder1978} in K-G operators treatment. Besides, some of Weder's theorems are useful to achieve the purposes of this paper, being presented always that necessary. 

This work is organized as follows. In section \ref{sec2}, it is presented a brief review on Klein-Gordon operator written in Schrodinger's form, as well as is constructed the one-dimensional analogous to the \textit{von} Neumann-Wigner potential. This potential is still used to demonstrate the existence of eigenvalues in the continuum spectrum to K-G operator. In section \ref{sec3}, the same methodology of previous section is extended to Coulomb interactions. In addition, the absence of eigenvalues in the continuum spectrum is demonstrated for some potentials. As it turns, we summarize our main findings and draw some perspectives in section \ref{sec4}. Throughout this work, we use units of $\hbar = c = 1.$

\section{Eigenvalues embedded in the continuum spectrum of Klein-Gordon operator \label{sec2}}

We begin this section with a brief review on K-G operators and Weder's conditions that prove the self-adjointness of K-G operator. Remembering that general expression to K-G equation is given by
\begin{equation}  
\label{s1.1}
\left( i \frac{\partial}{\partial t}-b_{0}\right)^{2}\psi(x,t)=\left[%
\sum_{i=1}^{n}\left(D_{i}-b_{i} \right)^{2}+m^{2}+q_{s} \right]\psi(x,t),
\end{equation}
where $x \in \mathbb{R}^{n}, t \in \mathbb{R}, D_{j}=-i \partial/\partial_{x_{j}},
b_{0}(x)$ is the electric potential, $b_{i}(x) (1\leq i \leq n$) is the magnetic potential, and $q_{s}(x)$ the scalar potential. By implementing the \textit{ansatz} transformations of Eq. \ref{s1.2} in Eq. \ref{s1.1}
\begin{equation}  
\label{s1.2}
f_{1}=\psi(x,t), \quad f_{2}=i\frac{\partial \psi}{\partial t}(x,t),
\end{equation}
it is possible written k-G equation given by expression \ref{s1.1} in its hamiltonian form with electromagnetic and scalar potentials
\begin{align}
& i\frac{\partial f}{\partial t}=hf, \\
& h=
\begin{bmatrix}
0 & 1 \\ 
l & Q
\end{bmatrix}
, \\
& l=\sum_{i=1}^{n}\left(D_{i} - b_{i}\right)^2+m^2+q(x), \\
& q(x)=q_{s}-q_{0}^{2}, Q(x)=2b_{0}, \\
& D(h)=\{f \in C_{0}^{\infty}(\mathbb{R}^{n})^{2}
, l f_{1} \in L^{2}(\mathbb{R}^{n}), Q f_{2} \in L^{2}(\mathbb{R}^{n})\}.
\label{s1.3}
\end{align}

In this case, the energy sesquilinear form associated can be written as

\begin{align} & \left(f,g\right)_{E}=\sum_{i=1}^{n} \left(
\left(D_{i}-b_{i}\right)f_{1},\left(D_{i}-b_{i} \right)g_{1}\right) \\ &
+\left(\left(m^{2}+q \right)f_{1},g_{1} \right)+\left(f_{2},g_{2} \right),
\\ & f,g \in C_{0}^{\infty}\left(\mathbb{R}^{n}\right).
\label{s1.4}
\end{align}

It is possible to verify that $h$ is symmetric in relation that sesquiliner form. In general, the value of sesquilinear form is not positive. To ensure its positivity, Weder \cite{Weder1977, Weder1978} introduced the following assumption:

\textbf{I)} There is a constant $\epsilon \geq 0$ such that
\begin{align}
& \int q^{-}|f|^{2}d^{n}x \leq \sum_{i=1}^{n}\| D_{i}f \|^{2}+\left(
m^{2}-\epsilon \right) \|f\|^{2}, \\
& f \in C_{0}^{\infty}\left( \mathbb{R}^{n}\right)  
\label{s1.5}
\end{align}
where $q^{\pm}$ are the positive and negative parts of the measurable function $q(x)$.

Notice that if \textbf{I)} holds sesquilinear form, $\left(\cdot , \cdot \right)_{E}$, starts to define a norm and $\textfrak{H}_E$ is defined as the completion of $C_{0}^{\infty}\left(\mathbb{R}^{n} \right)^{2}$ with this norm, as can seen in Ref. \cite{Weder1978} (Lemma 2.1). It is relevant to emphasize that for a pure electric potential interaction in the K-G equation $q_{s}=b_{i}=0$ and $b_{0}=\frac{e}{|x|}$ , \textbf{I)} holds if $|e| \leq (n-2)/2$.

So far we have discuss details of K-G operators as well as some requirement mathematical conditions related to density of probability in order to unaltered keep its positivity. Let us now start reviewing the problem of absence of positive eigenvalues for Scrhodinger operators.

We have know that there is an interesting potential in which the eigenvalues equation for Schordinger's operators has a positive eigenvalue embedded in the continuous spectrum with an unitary and positive value.This potential is namely \textit{von} Newmann-Wigner potential, and can be written as \cite{von1993merkwurdige, Simon1969}
\begin{align}
V_{NW}(r)=\frac{-32 \sin r \left[ \zeta (r)^{3} - 3 \zeta (r)^{2} \sin^{3} (r) + \zeta (r) r+ \sin^{3} r \right]}{(1+\zeta (r)^{2})^{2}},
\label{VN_pot}
\end{align}
where $\zeta(r)=2r-sin(2r)$, $|\textbb{x}|=r$ is the distance from the origin in $\mathbb{R}^{3}$, and $V_{NW}(r)$ is a physical example in which there is a square integrable $\psi$ with $(\Delta + V) \psi = \psi$ that provides an unitary eingenvalue. The mathematical computations to obtaining that potential can be extend to the one-dimensional case using $(\Delta + V) \psi = \psi$. A slightly different square integrable eigenfunction in $\mathbb{R}$ when compared to $\mathbb{R}^3$ case can be obtained by construction of the potential on the real line. In this situation, it is possible to obtain a similar expression to that given by Eq. \ref{VN_pot}, with the eigenfunction
\begin{equation}
\psi(x)=(sin(x))(1+\zeta(x)^{2})^{-1}.
\end{equation}
Notice that the contructed one-dimensional potential is differ to given by Eq. 13 only by argument $x \in \mathbb{R}$. In addition, this allows constructing a potential in which there is a self-adjoint K-G operator with positive eigenvalue in $(- \infty,-m] \bigcup [m, \infty )$.

In order to demonstrate the existence of a positive eigenvalue let us introduce some concepts and assumptions used by Weder \cite{Weder1978}:\begin{equation}
N_{\alpha , \delta}=sup_{x} \int_{|x-y|<\delta} |g(y)|^{2} \omega_{\alpha}(x-y) dy,
\end{equation}
where

\[
  \omega_{\alpha}(x)=\begin{cases}
               |x|^{\alpha -n}, \quad \alpha<n \\
               1-log|x|, \quad \alpha=n \\
               1, \quad \alpha>n,
            \end{cases}
\]

$N_{\alpha}(g) \equiv N_{\alpha ,1}(g)$ and $N_{\alpha}$ the set of all $g$ such that $N_{\alpha}(g)$ is bounded.

\textbf{II')} For $1 \leq i \leq n, b_{i}(x) \in N_{2}$, and if $n \geq 2, N_{2,s}(b_{i}) \rightarrow 0$ as $s \rightarrow 0$.

\textbf{III')} $|q(x)|^{1/2} \in $, and if $n \geq 2, N_{2,s}(q) \rightarrow 0$ as $s \rightarrow 0$.

\textbf{IV')} $C(x)=\sum_{i=1}^{n}(\frac{\partial}{\partial x_{i}}b_{i}(x)) \in N_{4}$, and if $n \geq 4, N_{4,s}(C) \rightarrow 0$ as $s \rightarrow 0$.

\textbf{V')} $q(x) \in N_{4}$, and if $n \geq 4, N_{4,s}(q) \rightarrow 0$ as $s \rightarrow 0$, $b^{2}=\sum_{i=1}^{n}(b_{i}(x))^{2} \in N_{4}$, and if $n \geq 4, N_{4,s}(b^{2}) \rightarrow 0$ as $s \rightarrow 0$.

\textbf{VI')} $b_{0}(x) \in N_{2}$, and if $n \geq 2, N_{2,s}(b^{0}) \rightarrow 0$ as $s \rightarrow 0$.

These set of conditions is relevant to correctly understand the follow theorems, as well as to achieve the central goal of this section, i. e., demonstrate the existence of positive eigenvalue in the continuous spectrum.

Allow us to write the theorem 2 from Weder's \cite{Weder1978} article as \textbf{Theorem 1}.
\begin{theorem}
If \textbf{I)} and \textbf{II')}-\textbf{VI')} are satisfied, $h$ has a self-adjoint extension, $H$, on $ \textfrak{H}_E$ with domain $D(H)=H_{2} \otimes H_{1}$.
\end{theorem}
Here $H_{k}$ ($k=1,2$) are standard Sobolev spaces $W^{k,2}$. We note that in some cases, a bounded $g$ implies $N_{\alpha}(g)< \infty$. 

So far, we demonstrate that the $h$ operator has a self-adjoint extension, used along of this section.

We have almost all tools to demonstrate the central theorem of this section, except the \textbf{Lemmas 1} and \textbf{2} that are presented follow. Particularly, the \textbf{Lemma  1} is interesting in the treatment of bounded potentials, and its proof is trivial.
\begin{lemma} 
\label{lemma1}
Let $g(y)$ be bounded and $\alpha \geq n$, then $N_{\alpha}(g)$ is bounded.
\end{lemma}

An alternative form to \textbf{I} is presented in \textbf{Lemma 2}, whereby the condition \textbf{I} is easly verified. 
\begin{lemma} 
\label{lemma2}
Let
\begin{equation}
S_{\lambda}(q)=sup_{x}\int |q(y)|\ G_{2, \lambda}(x-y)dy,
\end{equation}
where $G_{2, \lambda}(q)$ is the inverse Fourier transform of
\begin{equation}
(2 \pi)^{- \frac{n}{2}}(\lambda + |\eta|^{2})^{-1}, \quad \lambda \geq 0,
\end{equation}
then \textbf{I)} holds if and only if $S_{\lambda}(q^{-}) \leq 1$, for some $\lambda < m^{2}$.
\end{lemma}

A proof of \textbf{Lemma} \ref{lemma2} can be found in Ref. \cite{Weder1978}.

As adressed in introduction of this paper, a traditional problem of spectral theory of Srcrodinger operators is determine the existence, or no, of eigenvalues in the essential spectrum. However, when we take K-G operators such existence need to be better investigated.  In order to analyse the eventual existence of eigenvalues in the essential spectrum to K-G operators, we present the follows theorem where a pure scalar interation is taken into account. The \textbf{Theorem} \ref{counter} shows that if the scalar interaction is the \textit{von}-Newmam-Wigner potential there are eigenvalues embeeded in the esssential spectrum of the k-G operator. It is relevant to emphasize that this our result can be interpreted as the relativistic version of the demonstrations already introduced by \textit{von}-Newmman-Wignner to Schrodinger operators in Ref. \cite{von1993merkwurdige}. This result does not occurs to pure eletric interaction, as can seen in section \ref{sec3}.
\begin{theorem} \label{counter}
$h$ with a real valued pure scalar interaction
\begin{equation} \label{s-potential}
q_{s}=V_{WN},
\end{equation}
where $V_{WN}$ is Wigner - von Newmman potential, has a self-adjoint extension $H_{F}$, with $D(H_{F})=H_{2} \otimes H_{1}$, and further $H_{F}$ has eigenvalues embedded in $(-\infty,-m] \bigcup [m, + \infty)$.
\end{theorem}

\begin{proof}
Keep $n=1$ and the pure scalar interaction given by Eq. \ref{s-potential}, then \textbf{II')} and \textbf{IV')} are directly satisfied. By the fact that in this particular case $q=q_{s}$, and $V_{WN}(x)$ is bounded,  it is possible to see that $|q|^{1/2} \in N_{2}$ by applying \textbf{Lemma} \ref{lemma1}, and then we have \textbf{III')}. In an analogous way, \textbf{V')} and \textbf{VI')} hold as well. Let us suppose \textbf{I)} to be valid, and this will be really demonstrated in the end of this theorem, then, by Weder's theorem, $H_F$ is the necessary self-adjoint extension. By using of the eigenvalue problem $H_{F} \psi = E \psi$, one get
\begin{equation} \label{map}
(- \Delta_{1}+q_{s}+m^{2})\psi_{1}=E^{2}\psi_{1}, \quad \psi_{1} \in H_{2}(\mathbb{R}),
\end{equation}
where $\Delta_n$ denotes the Laplacian in $n$ dimensions. Pick $q=V_{WN}$, then we have $E= \pm \sqrt{1+m^{2}}$ due to the presence of \textit{von}-Neumann Wigner potential in (\ref{map}). It is important to stress that to the value  $E= - \sqrt{1+m^{2}}$ the \textbf{I} is not satisfied. Although this found negative value is in the known Dirac sea region, the presence of eigenvalues in this region should be better investigated. Now, let us consider \textbf{I)}. For the choice above of $q$, and by \textbf{Lemma} \ref{lemma2}, we see that \textbf{I)} holds if $S_{\lambda}(V_{WN}^{-}) \leq 1$. Indeed
\begin{equation}
\sup_{x} \int |q^{-}(y)|G_{2, \lambda}(x-y) dy = (1/2 \sqrt{\lambda}) \sup_{x} \int |V_{WN}^{-}| e^{- \sqrt{\lambda}|x-y|dy},
\end{equation}
where the "supression" of the integral ensures the existence of an infinite number of particular values for $\sqrt{\lambda}<m$, such that \textbf{I)} holds true. \rule{5pt}{5pt}
\end{proof}

Notice that the leading term of $V_{WN}$ for large $x$ is given by 
\begin{equation}
-\frac{8 \sin (2x)}{x},
\end{equation}
hence, $\limsup x (\partial V_{2})/\partial x = 16$, where $V_{2}$ is the infinitely differentiable function defined with more details in \textbf{Theorem 3} of section \ref{sec3}. Therefore, it is possible to conclude that $E^{2}-m^{2}<16$ in \textbf{Theorem} \ref{counter} (see the intitulated "main Theorem" in Ref. \cite{Simon1969}). Thus, even in this "counter-example case", one can concludes that there are no eigenvalues in $(\sqrt{16+m^{2}},\infty)$.

We would be have problems when applying Weder's theorem in the case $n \geq 2$ (see \textbf{Theorem} \ref{counter}). For example, even for $n=2$, we do not have \textbf{III')} completely satisfied since $N_{2,s}(q) \not\to 0$ as $s \rightarrow 0$ for $q$ (in the case under consideration). The fact we have used the pure scalar interaction is quite technical. Suppose we would like to find eigenvalues in $[m,+\infty)$, for $h$ with a pure electric interaction given by

\begin{equation} \label{map2}
-(b_{0})^2+2Eb_{0}=V_{WN}.
\end{equation}
Thus, even for the most suitable choice for $b_{0}$, we would find problems with the expressions involving the requirement \textbf{I)} with the map (see expression \ref{map2}), the stronger inequality
\begin{equation}
S(q_{1}+ V_{WN}^{-})\leq 1,
\end{equation}
where $q_{1}=1+m^{2}$ takes place.  So, \textbf{Theorem} \ref{counter} to pure electric interactions remains open for plausible values of the mass.

So far, we have demonstrated the existence of a positive eigenvalue in the region of spectrum similar to the found results for continous spectrum of non-relativistic Schorodinger operator. In next section, as already mentioned above, in contrast of \textbf{Theorem} \ref{counter}, we prove that K-G operator with pure eletric interactions does not has eigenvalues embeeded in the essential spectrum.  

\section{Absence of eigenvalues in "continous spectrum" region for Coulomb electric interaction \label{sec3}}

Now, we extend same methotodology of previous section in order to determine the absence of eigenvalues in the essential spectrum of K-G operators with Coulomb interactions.To do that, we list some important results about self-adjointness, and on the essential spectrum location, using a sequence of complementary assumptions introduced by Weder in Ref. \cite{Weder1978}. In short, these assumptions are related to the self-adjointness;  the essential spectrum location; and the absence of eigenvalues theorems concerning to the K-G operators. 

Weder's hypothesis (see Ref. \cite{Weder1978}) are listed in \textbf{I)}-\textbf{V)}. In addition, we present here assumption \textbf{VI)}. Note that hypotesis \textbf{I)} is already presented in section \ref{sec2}.

\textbf{II)} 

(1) $b_{i}(x) \in N_{2}$, $1 \leq i \leq n$, and if $n \geq 2$, $N_{2,\delta}(b_{i}) \xrightarrow{\delta \rightarrow 0} 0$, \\

(2) $q(x)=q_{1}(x)+q_{c}(x)$, $|q_{1}|^{1/2} \in N_{2}$, and if

$n \geq 2$, $N_{2,\delta}(|q_{1}|^{1/2})\xrightarrow{\delta \rightarrow } 0$ $\cdot$ $q_{c}(x) \equiv 0$ if $ n \leq 2$, and

$q_{c}(x)=-\frac{e^{2}}{|x|^{2}}, |e|(n-2)/2$ if $n>2.$

\textbf{III)}

$(N_{4,x}(b_{i})+N_{4,x}(|q_{1}|^{1/2})) \xrightarrow{\delta \rightarrow 0} 0$,

\textbf{IV)}

$b_{0}=b_{0}^{1}+b_{0}^{c} \in N_{2}$ and if $n \geq 2$

$N_{2,j}(b_{0}^{1})\xrightarrow{\delta \rightarrow } 0$.  $q_{c}(x) \equiv 0$ if $n \leq 2$.

If $n \geq 3$, $b_{0}^{c}(x)=\frac{e}{|x|}$, where $|e|<\frac{n-2}{2 \sqrt{17}}$.

\textbf{V)}

$N_{2,x}(b_{0}^{1})\xrightarrow{|x| \rightarrow \infty} 0$.

\textbf{VI)} $b_{0}^{c}(x)$ is such that, for some real $k$, $-b_{0}^{c}(x)^{2}+kb_{0}^{c}(x)$ is a real function in \textbf{$R^{n}$} so that:

(1) Multiplication by $V_{1}+V_{2}=-b_{0}^{c}(x)^{2}+kb_{0}^{c}(x)$ is $- \Delta$-bounded  with relative bound less than $1$.

(2) There is a closed set $S$ of measure zero so that \textbf{$R^{n}$}$\backslash S$ is connected and so that $V_{1}+V_{2}$ is bounded on any compact subset of \textbf{$R^{n}$}$\backslash S$.

(3) $-V_{1}$ is bounded outside some ball $\{x| |x|< r_{0}\}$, and $|x| b_{0}^{2}(x) \rightarrow 0$ as $|x| \rightarrow \infty$.

(4) $V_{2}$ is bounded outside some ball $\{x| |x|< r_{0}\}$, and $k b_{0}^{2}(x) \rightarrow 0$ as $|x| \rightarrow \infty$. (This assumption can be replaced by $k b_{0}(x)<0$ for $|x|>r_{0}$).

(5)  If $b_{0}^{c}$ maps $r$ in $b_{0}^{c}(r,\cdot)$ from $(0,\infty)$ to $L^{\infty}(S^{n-1})$ where $S^{n-1}$ is the $(n-1)$-dimensional sphere, then for $|x|=r>r_{0}, b_{0}^{c}(x)$ is differentiable as an $L^{\infty}$-valued function and for some suitable $k$, $k \overline{lim}_{r \rightarrow \infty} r \frac{\partial b_{0}^{c}}{\partial r} \leq 0$.

If \textbf{I)} and \textbf{II)} are satisfied, the norm of the completion $%
\textfrak{H}_E$ is equivalent to the norm of $H^{1}(\textbf{R}^{n}) \otimes L^{2}(\textbf{R}^{n})$, and they coincide as sets (this is the Lemma 1.4 by Weder \cite{Weder1978}). By this lemma and assumption \textbf{II)}, $l$ (remember Eq. (\ref{s1.3}) above) has a self-adjoint bounded below extension, denoted by $L$ (this is the Lemma 1.5 in Ref. \cite{Weder1978}). 

Now, let
\begin{align}
H=H_{L}+V, D(H)=D(L) \otimes H^{1} \\
H_{L}=
\begin{bmatrix}
0       & 1  \\
 L       & 0  \\
\end{bmatrix},
V=
\begin{bmatrix}
0 & 0 \\
 0 & Q \\
\end{bmatrix}
\nonumber
\end{align}
then if \textbf{I)} - \textbf{V)} are satisfied, then $H$ is a self-adjoint extension of $h$ see (\ref{s1.3}), with domain $D(H)=D(L) \otimes H^{1}(\textbf{R}^{n})$ and its essential spectrum coincides with $(-\infty,-m]\bigcup[m,\infty)$ (this reult is presented in theorem 2 in Ref. \cite{Weder1977}).

So far, using the formalism introduced by Weder in Refs. \cite{Weder1977} and \cite{Weder1978}, it is possible to obtain the \textbf{Theorem} \ref{main}. This result is immediately presented below, and opportunetely demonstrated after presenting an important appilication of that.

\begin{theorem} \label{main}
The self-adjoint extension $H$ of the K-G operator, with a pure electric field $b_{0}^{c}$ satisfying \textbf{I)} - \textbf{VI)} (we call this operator $\tilde{H}$), has no eigenvalues embedded in its essential spectrum.
\end{theorem}

As an application of \textbf{Theorem} \ref{main}, we can prove the $\textbf{R}^{3}$ case, in which includes the non-central electric Coulomb interaction. Even though the following result is stated here as the \textbf{Theorem} \ref{R3}, it is a corollary of \textbf{Theorem} \ref{main}. Notice that for the important case $S=\{\textbf{0}\}$, we don't have the necessary connectedness present in assumption \textbf{VI)} - $2)$, and in this case we shall directly apply other \textbf{Theorem} (see details in Ref. \cite{Simon1969}) with appropriate modifications, presented here as the \textbf{Theorem 5} below.

\begin{theorem} \label{R3}
The self-adjoint extension, $H$, of the K-G operator with a pure $ b_{0}^{c}$ as a Coulomb interaction, $\frac{e}{|x|}$, in $\textbf{R}^{3}$ (called here $\tilde{H}_{C}$), has no eigenvalues embedded in its essential spectrum in the following sense: for the attractive case ($b_{0}^{c}<0$) this is valid for $[m,\infty)$, and for the repulsive ($b_{0}^{c}>0$), it holds in $(-\infty,-m]$.
\end{theorem}

\subsection*{Proof of the Theorem \ref{main}:}

We can use to K-G operator, and so demonstrate \textbf{Theorem} \ref{main}, the same technique already used to Schrödinger's operator (as in \textbf{Theorem} \ref{counter}), where is necessary to written K-G operator in Schodinger's form. On the other hand, it is relevant to observe some differences. In this case, see that we are using pure electric interactions, and additional cares must to be taken during the transition process from K-G to Schrödinger equation. In fact, the relation between Schr{\"o}dinger and K-G operators is such that it is possible to construct a map between particular effective Sch{\"o}dinger potentials, and pure electric potentials in the structure of the K-G operator. 
Although such map includes an explicit dependence on the spectral parameter, it is important to emphasize that all the computations can be perfectly done by keeping then fixed. More details about this approach can be found in Ref. \cite{Lundberg1973}.

We have that the self-adjoint extension, $\tilde{H}$, of the K-G operator without scalar potential, magnetic potential and $b_{0}^{1}=0$, becomes
\begin{align*}
\tilde{H}=
\begin{bmatrix}
    0       & 1  \\
   -\Delta+m^{2}-(b_{0}^{c})^{2}       & 0
\end{bmatrix} \\
D(\tilde{H})=D(L) \otimes H^{1}.
\end{align*}

It is clear that the absence of magnetic and scalar potentials, and the choice $b_{0}^{1}=0$, are all in accordance with Weder's assumptions \textbf{I)}  - \textbf{V)}. Then, we write the eigenvalue equation $\tilde{H} \psi=E \psi$,
\begin{align*}
\psi=
\begin{bmatrix}
    \psi_{1} \\
    \psi_{2}
\end{bmatrix} \in D(\tilde{H}).
\end{align*}

Hence, in terms of $\psi_{1}$, one can see the compatibility of the above K-G eigenvalue equation with the Schr{\"o}dinger type one,

\begin{equation*}
(-\Delta + V_{eff})\psi_{1}=\tilde{E}\psi_{1}, \quad V_{eff} \equiv -(b_{0}^{c})^(2)+2Eb_{0}^{c}, \quad \tilde{E} \equiv E^{2}-m^{2}.
\end{equation*}
Pick $V_{1}(x)=-(b_{0}^{c})^{2}$ and $V_{2}(x)=2Eb_{0}^{c}$. See that by \textbf{VI)}, $D(L)=H^{2}$, which is the domain of the free Schrödinger operator, and this is enough to the use of known Kato's theorem (see more details in Ref. \cite{Kato1995}). Hence, $V_{eff}$ satisfies all the hypothesis of Kato-Agmon-Simon's theorem \cite{ReedSimon4}. Thus, the \textbf{Theorem} \ref{main} is proven. \rule{5pt}{5pt}

\subsection*{Proof of the Theorem \ref{R3}:}

As mentioned, the proof of \textbf{Theorem} \ref{R3} is based in the use of a theorem demonstrated in Ref. \cite{Simon1969}, with conditions sligtly modified without loss of generality, in order to adapt it to our case. Here, we enunciate it as \textbf{Theorem} \ref{simon}.

\begin{theorem}{} \label{simon}
Let $V$ be a real valued function in \textbf{$R^{n}$}$\backslash \textbf{0}$ with the following properties:

(a) $V \in L^{2}+L^{\infty}$;

$V=V_{1}+V_{2}$ with

(b) for some $R_{0}$, $V_{1}$ and $V_{2}$ are $C^{3}$ in $M=\{\textbf{r}|r>R_{0}\}$,

(c) $\lim_{r \rightarrow \infty}V_{1}(\textbf{r})=0$,

(d) for $r>R_{0}$, $V_{2}(\textbf{r})<0$,

(e) for $r>R_{0}$, $-\frac{\partial V_{2}}{\partial r} \leq -\frac{1}{r}V_{2}$.

Then $-\Delta+V$ has no eigenvalues in $(0,\infty)$.
\end{theorem}

We have that the self-adjoint extension of the K-G operator $\tilde{H}_{c}$, without scalar and magnetic potentials, $b_{0}^{1}=0$, and $q_{c}(x)=-\frac{e}{|x|^{2}}$, $|e| \leq (n-2)/2$ (if $n>2$), and $b_{0}^{c}(x)=\frac{e}{|x|}$, $|e| \leq (n-2)/2 \sqrt{17}$ (if $n>3$), becomes

\begin{align*}
 \tilde{H}_{C}=
\begin{bmatrix}
   0       & 1  \\
   -\Delta+m^{2}-\frac{e^{2}}{|x|^{2}}       & \frac{2e}{|x|}
\end{bmatrix} \\
D(\tilde{H}_{C})=D(L) \otimes H^{1}.
\end{align*}

It is easy to see that the lack of magnetic and scalar potentials, and the choice $b_{0}^{1}=0$, are turning \textbf{I)}-\textbf{V)} conditions satisfied. Hence, we get the eigenvalue equation for $\tilde{H}_{C}$, $\tilde{H}_{C} \psi=E \psi$,

\begin{align*}
\psi=
\begin{bmatrix}
    \psi_{1} \\
    \psi_{2}
\end{bmatrix} \in D(\tilde{H}_{C}).
\end{align*}
Then, in terms of the component $\psi{1}\in H^{2}$, one can see the compatibility of the above K-G eigenvalue equation with the following Schrödinger equation,

\begin{equation*}
(-\Delta + V_{eff})\psi_{1}=\tilde{E}\psi_{1}, \quad V_{eff} \equiv -\frac{e^{2}}{|x|^{2}} + 2E\frac{e}{|x|}, \quad \tilde{E} \equiv E^{2}-m^{2}.
\end{equation*}

In this case, pick $V_{1}(x)=-\frac{e^{2}}{|x|^{2}}$ and $V_{2}=2E\frac{e}{|x|}$, hence $V_{eff}$ holds for all hypothesis of Simon's theorem or its modification given by \textbf{Theorem 5}, which means that the components $\psi_{1} \in H_{2}$ are ensuring the use of Kato's theorem, and so it is possible directly to apply the \textbf{Theorem 5}, with exception that for the attractive case, (d) holds only if $E>m$, and that for the repulsive case it holds only if  $E<m$. Thus, the Theorem \ref{R3} is proven. \rule{5pt}{5pt}

\section{Conclusions and Remarks\label{sec4}}

We found a class of physical potentials in which there is no eigenvalues embedded in the continous region of relativistic spectrum. It is possible to see that the fact associated to the non-existence of eigenvalues embedded in the essential spectrum can be intepreted as an expected result. We clearly show that with the correct location of continuous spectrum together with the result of theorem 3, we can rigour confirm if thresholds  $\pm m$ belong, or no, to the point spectrum. Here, we emphasize that essential spectrum coincides with $(-\infty,-m]\bigcup[m,\infty)$, and so $\pm m$ dont belong to the point spectrum. Generally, these results are not so clear by use of the computational calculations of theoretical physics.

It is worthy of emphasis that the present paper demonstrates that relativistic K-G operator contains eigenvalues in its essential spectrum, as occurs to Schrodinger's operators. This is a new and interesting result, undemonstrated before in literature to relativistic operators, and introduced here. It is remarkable that the potential in wich this occurs is \textit{von}-Neumann Wigner potential, as expected.  In particular, this potential provides similar results for  Schodinger's operators, i. e., the existence of an unitary eigenvalue in the essential spectrum.

We have expect to be possible extend the same methodology of theorem 2 in order to investigate the existence of eigenvalues in the essential spectrum for K-G operator to higher-dimensional order. However, in this case, the self-adjointness properties must be better investigated. 

\section{Acknowledgments}

The authors gratefully acknowledge the support provided by Brazilian Agencies CAPES, CNPQ. We would like to thank the following for their kind support: IFPI and USP; for the great professors D. M. Gitman, C. R. Oliveira and R. Weder for our private communications.

\bibliographystyle{unsrt}  
\bibliography{references}  
%%% Remove comment to use the external .bib file (using bibtex).

%\begin{thebibliography}{1}
%
%\bibitem{kour2014real}
%George Kour and Raid Saabne.
%\newblock Real-time segmentation of on-line handwritten arabic script.
%\newblock In {\em Frontiers in Handwriting Recognition (ICFHR), 2014 14th
%  International Conference on}, pages 417--422. IEEE, 2014.
%
%\bibitem{kour2014fast}
%George Kour and Raid Saabne.
%\newblock Fast classification of handwritten on-line arabic characters.
%\newblock In {\em Soft Computing and Pattern Recognition (SoCPaR), 2014 6th
%  International Conference of}, pages 312--318. IEEE, 2014.
%
%\bibitem{hadash2018estimate}
%Guy Hadash, Einat Kermany, Boaz Carmeli, Ofer Lavi, George Kour, and Alon
%  Jacovi.
%\newblock Estimate and replace: A novel approach to integrating deep neural
%  networks with existing applications.
%\newblock {\em arXiv preprint arXiv:1804.09028}, 2018.

%\end{thebibliography}

\end{document}